\documentclass[11pt]{llncs}
\usepackage{amsmath,amssymb,amsfonts,amsthm}
\usepackage{array}
\usepackage{float}
\usepackage{enumerate}
\usepackage{ifthen}
\usepackage{xspace}
\usepackage{graphicx,url}
\usepackage{fullpage}
\usepackage[a4paper]{geometry}
\usepackage{wasysym}
\usepackage{algorithm}
\usepackage{algpseudocode}

\usepackage{graphicx, subfigure}
\usepackage{epsfig}

\usepackage{setspace}
\newenvironment{theorem-repeat}[1]{
\setcounter{theorem}{\ref{#1}}
\addtocounter{theorem}{-1}
\begin{theorem}}
{\end{theorem}}

\newenvironment{lemma-repeat}[1]{
\setcounter{lemma}{\ref{#1}}
\addtocounter{lemma}{-1}
\begin{lemma}}
{\end{lemma}}

\newcounter{linecounter}
\newcommand{\linenumbering}{(\arabic{linecounter})}
\renewcommand{\line}[1]{\refstepcounter{linecounter}
\label{#1}
\linenumbering}
\newcommand{\resetline}{\setcounter{linecounter}{0}}

\newcommand{\remove}[1]{}

\begin{document}


\title{Robot Networks with Homonyms: The Case of Patterns Formation}

\author {
Zohir Bouzid
  \and 
Anissa Lamani
}

\institute{University Pierre et Marie Curie - Paris 6, LIP6-CNRS
  7606, France. 
  \email{zohir.bouzid@lip6.fr}
  \and
  University of Picardie Jules Verne, Amiens, France.
  \email{anissa.lamani@lip6.Fr}
}

\maketitle

\begin{abstract}
In this paper, we consider the problem of formation of a series of geometric patterns \cite{DFSY10} by a network of oblivious mobile robots that communicate only through vision. 
So far, the problem has been studied in models where robots are either assumed to have distinct identifiers or to be completely anonymous.
To generalize these results and to better understand how anonymity affects the computational power of robots, we study the problem in a new model, introduced recently in \cite{DFGKR11}, in which n robots may share up to $1 \leq h \leq n$ different identifiers.
We present necessary and sufficient conditions, relating symmetricity and homonymy, that makes the problem solvable.
We also show that in the case where $h=n$, making the identifiers of robots invisible does not limit their computational power. 
This contradicts a result of \cite{DFSY10}.
To present our algorithms, we use a function that computes the Weber point for many regular and symmetric configurations. This function is interesting in its own right, since the problem of finding Weber points has been solved up to now for only few other patterns.
\end{abstract}

\section{Introduction}

Robot networks \cite{SY99} is an area of distributed computing in which the object of the study
is the positional (or spacial) communication paradigm \cite{agmon2004ftg}:
robots are devoid of any means of \emph{direct} communication; 
instead, they communicate \emph{indirectly} through their movements and the 
observation of the current positions of their peers.
In most of the studies, robots are oblivious, \emph{i.e.}
without any memory about their past observations, computations 
and movements.
Hence, as with communication, memory is \emph{indirect} in some sense, it
is collective and spacial.
Since any algorithm must use a kind of memory, resolving problems in the context of 
robot networks is the art of making them remember without memory \cite{FIPS10}.
But this indirect memory has its limits, and one of the main goals of research in this
field is to precisely characterize the limiting power of obliviousness.

The problem of \emph{formation of series of patterns} (or patterns formation) is perhaps the best abstraction that captures the need of robots 
to have a form of memory even when the model does not provide it directly. 
In this problem, introduced by Das et al. in \cite{DFSY10}, 
robots are required to form periodically a sequence of geometric patterns $S=\langle P_1, P_2, \ldots, P_m \rangle$. 
At each instant,
robots should be able to know, just by observing the environment, which patterns they are about to form.
In \cite{DFSY10}, the authors study the problem in both anonymous and eponymous systems.
In particular, they show that the only formable series when robots are anonymous are those in which all the patterns of the series 
have the same
symmetricity.
In contrast, if robots are endowed with identifiers, nearly all possible series are formable.

The gap between the two worlds, one in which almost nothing is possible and another one where everything is possible, illustrates
how much anonymity can be a limiting factor when we are interested to computability.
This brings us to ask the following question: 
is there a possible model where robots are neither eponymous nor
completely anonymous ? and if so, what are the possible series we are able to form ?
In a recent paper \cite{DFGKR11}, Delporte et al.
introduce a new model of ``partial anonymity'' which they call \emph{homonymy} and they apply it to study 
Byzantine agreement.
In this model, the number of distinct identifiers in the system is given by a parameter $h$ 
which may take any value between $1$ and $n$.
In the current paper, we inject this notion of homonymy to the Suzuki-Yamashita model \cite{SY99,DFSY10} 
which creates a new model that we call \emph{robot networks with homonyms}.

Studying the problem of patterns formation in this context allows us to get a better insight
on how the combined effect of anonymity and obliviousness affects the computational power of robots.
We consider series of patterns where all robots are located in \emph{distinct} positions and we assume
that identifiers are \emph{invisible}.
Our main result is to prove that for a series $S=\langle P_1, P_2, \ldots, P_t \rangle$ to be formable,
it is necessary and sufficient that the number
of labels $h$ to be strictly greater than $\frac{n}{\textsc{sncd}(sym(P_i), sym(P_{i+1}))}$ where 
$P_i, P_{i+1}$ are any two successive patterns of $S$, $sym(P)$ denotes the symmetricity of $P$ and
$\textsc{sncd}(x, y)$ is equal to the smallest divisor of $x$ that does not divide $y$ if any, $n+1$ otherwise.

To present our algorithms, we use a function that computes the Weber point \cite{W37} for many regular configurations, 
\emph{i.e.} all those in which 
there is a sort of rotational symmetry around a point. 
Given a point set $P$, the Weber point $c$ 
minimizes $\sum_{r \in P} distance(x, r)$ over all points $x$ in the plane.
Our result may be interesting in its own right, since the problem of finding Weber points has been solved up to now for only few other patterns 
(e.g. regular polygon \cite{ACP05}, a line \cite{CT90}). 

Finally, we consider the case where $(n=h=3)$ and where multiplicity points (in which many robots are located) are allowed.
In this setting, we prove that robots are able to form any series of patterns, contradicting a result of \cite{DFSY10}.
This has an interesting consequence: 
it means that making the identifiers of robots invisible does not
limit their computational power in the considered model, \emph{i.e.} they can form the same series of patterns as robots with visible identifiers.

\paragraph*{Roadmap} The paper is made up of six sections. In Section \ref{hom:sec:model} we describe the computation model
we consider in this paper.
Section \ref{hom:sec:symmetricity} defines the important notions of symmetricity, regularity and their relations to Weber points.
In Section \label{hom:sec:regular} our algorithms for the detection of Weber points are presented.
Then, we present in Sections \ref{hom:sec:formationImpossibility} and \ref{hom:sec:formationPossibility} the necessary 
and sufficient 
conditions that makes series of geometric patterns formable in our model.
In Section \label{hom:sec:threeRobots}, 
we address the special case where robots have distinct identifiers and where multiplicity points are allowed.
Finally, we conclude this paper in Section \ref{hom:sec:conclusion}

\section{Model}
\label{hom:sec:model}
Our model is based on a variation of the ATOM model \cite{SY99}
used in \cite{DFSY10},  to which we introduce the notion of homonymy \cite{DFGKR11}.
The system is made up of $n$ mobile robots $r_1, \ldots, r_n$ that communicate only through vision.
That is, robots are devoid of any mean of direct communication, the only way of them to communicate is by observing the positions of their peers (a ``read'')
and by moving in the plane (a ``write'').
To each robot is assigned an identifier (or label) taken from a set of $h$ distinct identifiers $\{1, \ldots, h\}$.
The parameter $1\leq h \leq n$ is called the \emph{homonymy} of the system.
The identifier of robot $r$ is denoted by $label(r)$. 
When two robots have the same label we say that they are \emph{homonymous}.
Robots are \emph{oblivious} in the sense that they do not have any memory of their past observations, computations and movements.
Hence, their actions are based entirely on their currently observed configuration and their label.
Each robot is viewed as a \emph{point} in a plane, thus multiple robots may lie in the same location forming a \emph{multiplicity} point.
We consider multiplicity points only in Section \ref{hom:sec:threeRobots}. 
There, 
we assume that robots are able to know exactly how much robots are located in each point.
This capability is known as the global strong \emph{multiplicity detection}. 
We assume also that robots do not obstruct the vision of each others.   
Robots are \emph{disoriented}, that is, each robot has its own local coordinate system with its own origin, axis, and unit of length which may change
at each new activation.
However, as in \cite{DFSY10},
we assume that robots share the same notion of clockwise direction, we say that they have the same \emph{chirality}. 


The execution unfolds in \emph{atomic} cycles of three phases Look, Compute and Move. 
During the Look phase, an activated robot take a snapshot of the environment using its visual sensors.
Then it calculates a destination in the Compute phase. The chosen destination is based solely on the label of the robot and the snapshot 
it obtained in the preceding phase.
Finally, 
In the Move phase, the robot \emph{jumps} to its destination.
A subset of robots are chosen for execution (activated) at each cycle by a fictional external entity called a \textit{scheduler}.
We require the scheduler to be fair \textit{i.e.} each robot is activated infinitely often.
Robots that are activated at the same cycle execute their actions synchronously and atomically.

\paragraph*{Notations}
Given to distinct points in the plane $x$ and $y$, 
$|x,y|$ denotes the Euclidean distance between $x$ and $y$,
$[x,y]$ the line segment between them 
and $(x, y)$ is the line that passes through both of them.
Given a third point $c$, $\sphericalangle(x,c,y, \rightturn)$ denotes the clockwise angle at $c$ from $x$ to $y$.
$\sphericalangle(x,c,y, \leftturn)$ is defined similarly.
When the information about orientation can be understood from context, we may omit the parameters $\rightturn$ and $\leftturn$.
Note that $\sphericalangle(x,c,y, \rightturn)$ denotes both the angle and its size, the difference between 
the two notions should be clear from context. 
Given a configuration $P$, its
smallest enclosing circle,
denoted $\textsc{SEC}(P)$,
is the circle of minimum diameter 
such that every point in $P$ is either on or in the interior of this circle.
Given a circle $C$, $rad(C)$ and $diam(C)$ denote 
its radius and diameter respectively. A point $p$ belong to $C$, written $p\in C$ if it on or inside $C$.
If $S$ is a set, we denote its cardinality by $|S|$.

\paragraph*{Problem Definition} \cite{DFSY10}
A configuration $C$ of $n$ robots is a multiset of $n$ elements $\{(p_1, l_1), \ldots (p_n, l_n)\}$ where $p_i$ is the position of robot $r_i$ 
and $l_i$ is its label. The set of positions $L(C)$ is the set of points occupied by at least one robot in $C$.
A pattern $P$ is represented by a set of $n$ distinct points $\{p_1, \ldots, p_n\}$. Two patterns are said isomorphic if they 
can be obtained from each others by way of translation, rotation and uniform scaling. 
$size(P)$ is the cardinality of $P$.
We say that a system of robots has formed the pattern $P$ 
is the set of of points of the current configuration $L(C)$ is isomorphic to $P$.
\textbf{The Problem of Formation of Series of Patterns} is defined as follows. 
As input, robots are given a periodic series of patterns
$\langle P_1, P_2, \ldots ,P_{m} \rangle ^{\infty}$. 
It is required that $\forall \tau: \forall P_i \in S: \exists \tau_i:$ 
robots form $P_i$ at $\tau_i$.
\section{Symmetricity, Regularity and Weber Points}
\label{hom:sec:symmetricity}
In this section we formally define two metrics that quantify how much configurations of \emph{distinct} 
robots may be symmetric: the strongest one, called symmetricity and 
the weaker one, called regularity. 
We then explain how these two notions relate to Weber points.
We start by defining a polar coordinate system which we use to state our definitions and algorithms.

\subsection{Polar Coordinate System}
\label{hom:subsec:polarCoordinate}
In the remaining of the paper,
we express the locations of robots using a local \emph{polar} coordinate system based on the SEC of the current configuration 
and the position of the local robot\cite{K05}.
The \emph{center} $c$ of the coordinate system is common to all robots and it coincides with the center of the SEC.
In contrast, the \emph{unit of measure} is local to each robot $r$ and 
its definition depends on whether the robot is located on $c$ or not.
In the first case, the point $(1,0)$ of the local coordinate system of $r$
is any point in the plane that is not occupied by a robot. In the second case,
the point $(1,0)$ is the current location of $r$.
The positive common \emph{clockwise orientation} is provided by the underlying model.
Note that $(d, \theta)$ denotes the point located at distance $d$ from $c$ and angle $\theta$.

\subsection{Symmetricity}

\begin{definition}[View]
The \emph{view} of robot $r$, denoted $\mathcal{V}(r)$, is the current configuration expressed 
using the local polar coordinate system of $r$.
\end{definition}

Notice that since robots are located in distinct positions,
our definition of the view implies that if some robot is located at $c$,
then its view is unique, \emph{i.e.} it cannot be equal to the view of another robot.

Now we define the following equivalence relation on robots based on their views.

\begin{definition}[$\backsim$]
Given a configuration $P$, given any two robots $r, r^\prime \in P$:
\begin{center}
$(r \backsim r^\prime) \Leftrightarrow (\mathcal{V}(r)=\mathcal{V}(r^\prime))$
\end{center}
The equivalence class of $r$ is denoted by $[r]$.
\end{definition}

\begin{property}\cite{DFSY10}
\label{hom:prop:convexPolygon}
Let $r\in P$.
If $|[r]|>2$, $[r]$ is a set of robots located at the vertices of
a convex regular polygon with $|[r]|$ sides whose center $c$ is the center of $SEC(P)$.
\end{property}

\begin{definition}[Symmetricity]
\label{hom:def:symmetricity}
The \emph{symmetricity} of a configuration $P$, denoted $sym(P)$, is the cardinality of the smallest 
equivalence class defined by $\backsim$ on $P$. 
That is, $sym(P)=min\{|[r]|~|~ r \in P\}$.
If $sym(P)=m$, we say that $P$ is $m$-symmetric.
\end{definition}

Note that despite the fact that our definition of symmetricity is different from the one used in \cite{DFSY10},
it is still equivalent to it when configurations does not contain multiplicity points, the case in which we are interested in the current paper.

The next lemma follows from our definition of symmetricity and the 
fact that the view of a robot located at $c$ is unique.

\begin{lemma}
\label{hom:lem:notLocatedInCenter}
If a configuration $P$ of distinct robots is $m$-symmetric with $m>1$, then
no robot of $P$ is located in $c$.
\end{lemma}

In the following lemma we prove that if no robot is located at the center of $SEC(P)$, then all the equivalence classes 
have the same cardinality.

\begin{lemma}
\label{hom:lem:sizeEquivalenceClasses}
Let $P$ be a configuration of $n$ distinct robots, and let $c$ be the center $SEC(P)$.
If no robot in $P$ is located at $c$, it holds that:
$$
\forall r, r^\prime \in P: |[r]| = |[r^\prime]|=m
$$
\end{lemma}

\begin{proof}
We assume that no robot is located in $c$ and we prove the following equivalent claim.
\begin{center}
$\forall r \in P: \forall r^\prime \in P: |[r^\prime]| \geq |[r]|$.
\end{center}

Fix $r\in P$ and let $x=|[r]|$.
If $x=1$, the claim holds trivially. Hence we assume in the following  that $|x| >1$.
Let $[r]=\{r_1, \ldots, r_x\}$.  
The indices define a polar ordering of the robots around $c$ 
and their addition is done $(mod~n)+1$.
It holds that $r^\prime \in \sphericalangle(r_i, c, r_{i+1})$ for some $i \in \{1, \ldots, x\}$.
But $\mathcal{V}(r_{i})=\mathcal{V}(r_{i+1})$ which means that there 
is a rotational symmetry of angle $\sphericalangle(r_i, c, r_{i+1})$ around $c$. 
Hence, there exists a robot $r^{\prime\prime}$ such that 
$r^{\prime\prime} \in \sphericalangle(r_{i+1}, c, r_{i+2})$ and $r^{\prime\prime} \backsim r^\prime$. 
By repeating this argument, we find $x$ robots that are equivalent to $r^\prime$ (including itself).
Hence $|[r^\prime]| \geq x=|[r]|$. This proves the lemma.
\end{proof}

\begin{lemma}
\label{hom:lem:setConcentricCircles}
Let $P$ be a $m$-symmetric configuration of $n$ distinct points with center (of symmetricity) $c$ and with $m>1$.
There exists a partition of $P$ into $x=n/m$ subsets $S_1, \ldots, S_x$ such that
each of them is a convex regular polygon of $m$ sides with center $c$.
\end{lemma}

\begin{proof}
$(m>1)$ implies that no point of $P$ is located in $c$ (by Lemma \ref{hom:lem:notLocatedInCenter}).
Hence, we can apply Lemma \ref{hom:lem:sizeEquivalenceClasses} and we deduce the existence of 
$x=n/m$ equivalence classes in $P\/\backsim$ with equal cardinality ($m$).
Denote them by $S_1, \ldots, S_x$.
According to Property \ref{hom:prop:convexPolygon}, each of these subsets consists in 
set of robots located at the vertices of
a convex regular polygon of $m$ sides with center $c$.
\end{proof}

The following lemma will be used by our algorithm for series formation in Section \ref{hom:sec:formationPossibility}
to break the symmetricity of configurations. 

\begin{lemma}
\label{hom:lem:2robotsEquivDistinctLabels}
Let $P$ be a $m$-symmetric configuration of $n$ distinct points with $m>1$.
If $h > n/m$, there exists two robots $r$ and $r^\prime$ such that
$\mathcal{V}(r) = \mathcal{V}(r^\prime)$ and $label(r) \neq label(r^\prime)$.
\end{lemma}

\begin{proof}
Assume not towards contradiction. 
This implies that all robots that have the same view in $P$ have the same label.
Note that since $P$ is $m$-symmetric, we have exactly 
$n/m$ equivalence classes (Lemma \ref{hom:lem:sizeEquivalenceClasses}) of equal cardinality. 
As all robots belonging to the same equivalence class are assumed to have the same label, 
we have $h=n/m$.
But we assumed $h>n/m$, contradiction.
\end{proof}



\subsection{Regularity}
In this section we formally define a weaker form of symmetry that we call regularity.
In the next section we state its precise relation with symmetricity.
We start by giving some useful definitions:

\begin{definition}[Successor]
\label{hom:def:successor}
Given a set $P$ of distinct points in the plane, and $c \not\in P$ a fixed point. 
Given some polar ordering of points in $P$ around $c$.
Let us denote by $S(r, c, \rightturn)$ the 
the clockwise successor of $r$ according
to the clockwise polar ordering of points $P$ around $c$.
The anticlockwise successor of $r$, denoted
$S(r, c, \leftturn)$, is defined similarly.

Formally \cite{K05}, $r^\prime=S(r, c, \rightturn)$ is a point of $P$ distinct from $r$ such that:
\begin{itemize}
\item If $r$ is not the only point of $P$ that lies in $[c,r]$, we take $r^\prime$
to be the point in $P \cup [c,r]$ that minimizes $|r, r^\prime|$.

\item Otherwise, we take $r^\prime$ such that 
no other point of $P$ is inside
$\sphericalangle(r,c,r^\prime, \rightturn)$. If there are many such points, we 
choose the one that is further from $c$.
\end{itemize}

When the center $c$ and the clockwise orientation are clear from context or meaningless, we simply write
$S(r, c)$, $S(r, \leftturn)$ or $S(q)$ to refer to the successor of $r$.
\end{definition}

\begin{definition}[$k$-th Successor]
\label{hom:def:ksuccessor}
The $k$-th clockwise successor of $r$ around $c$, denoted $S^k(r, c, \rightturn)$ (or simply $S^k(q)$) , is defined recursively as follows:
\begin{itemize}
\item $S^0(r)=r$ and $S^1(r)=S(r)$.
\item If $k>1$, $S^k(r)=S(S^{k-1}(r))$.
\end{itemize}

\end{definition}

\begin{definition}[String of angles]
Let $P$ be a set of distinct points in the plane, and $c \not\in P$ a fixed point.
The (clockwise) string of angles of center $c$ in $r$, denoted by $SA(r, c, \rightturn)$
is the string $\alpha_1 \ldots \alpha_n$ such that
$\alpha_i=\sphericalangle(S^{i-1}(r), c, S^i(r), \rightturn)$. 
The anticlockwise string of angles, $SA(r, c, \leftturn)$, is defined in a symmetric way.
Again, when the information about the center of the string $c$ and/or its clockwise orientation is
clear from context, we simply write $SA(r, c)$, $SA(r, \leftturn)$ or $SA(r)$.
\end{definition}

\begin{definition}[Periodicity of a string]
A string $x$ is $k$-periodic if it can be written as $x=w^k$ where $1 \leq k\leq n/2 $.
The greatest $k$ for which $x$ is $k$-periodic is called \emph{the periodicity} of $x$ and is denoted
by $per(x)$.
\end{definition}

The following property states that the periodicity of the string of angles does not depend
on the process in which it is started nor on the clockwise or anticlockwise orientation.

\begin{lemma}
Let $P$ be a set of distinct points in the plane, and let $c \not\in P$ be a point in the plane.
The following property holds:

$$
\exists m: \forall r \in P: \forall d \in \{\leftturn, \rightturn\}: per(SA(r, c, d))=m
$$
\end{lemma}

\begin{proof}
Follows from the fact that a periodicity of a string is invariant under symmetry and rotation.
\end{proof}

\begin{theorem}
\label{hom:thm:stringAngles}
Given is a set of points $P$, a center $c \not\in P$.
There exists an algorithm with running time $O(n ~logn)$ that computes $SA(c)$.
\end{theorem}

\begin{proof}
Fix some $r \in P$. We show how to compute $SA(r, c, \rightturn)$. 
The other cases are similar. We use an array $T[n]$ of $n$ cells. As a first step, we compute for each $p_i \in P$ the 
angle $\sphericalangle(r, c, p_i, \rightturn)$ and put the result in $T[i]$. Note that this step takes $O(n)$ time.
Then we sort $T$ in increasing order ($O(n~logn)$ steps). 
$SA(r, c, \rightturn)$ is the string $\alpha_1 \ldots \alpha_n$ such that
$\forall i \in [1, n-1]: \alpha_i=T[i+1]-T[i]$ and $\alpha_n=2\Pi-T[n]$.
The whole algorithm runs in $O(n~logn)$ time.
\end{proof}

The lemma means that when it comes to periodicity, the important information about a string of angles is only
its center $c$. Hence, in the following we may refer to it by writing $SA(c)$.

\begin{definition}[Regularity]
Let $P$ be a set of $n$ distinct points in the plane.
$P$ is $m$-regular (or regular) if there exists a point $c  \not\in P $ such that
$per(SA(c))=m>1$.
In this case, the \emph{regularity} of $P$, denoted $reg(P)$, is equal to $per(SA(c))$.
Otherwise, it is equal to $1$.
The point $c$ is called the \emph{center} of regularity.
\end{definition}

\begin{theorem}
\label{hom:thm:algoDetectRegularity}
Given is a set of points $P$, a center $c \not\in P$.
There exists an algorithm with running time $O(n ~logn)$ that detects if $c$ is a center of regularity for $P$.
\end{theorem}

\begin{proof}
The algorithm computes $SA(c)$. 
As proved in Theorem \ref{hom:thm:stringAngles}, this can be done in  $O(n ~logn)$ time.
Then it computes the periodicity of the string in linear time \cite{KMP77}.
If the obtained periodicity is greater than 1, it outputs \textsc{Yes}, otherwise it outputs \textsc{No}.
\end{proof}

\subsection{Weber Points}
In this section we state some relations between symmetricity and regularity.
Then we prove that their centers are necessarily Weber points, hence unique when the configuration is not linear.
The following lemma is trivial, its proof is left to the reader. It states the fact that if a configuration 
is $m$-symmetric, it is necessarily $m$-regular.

\begin{lemma}
\label{hom:lem:regEqSym}
Let $P$ be any configuration with $sym(P) >1$. It holds that $reg(P)=sym(P)$ and the center of regularity of $P$ 
coincides with its center of symmetricity.
\end{lemma}

The next lemma shows in what way the regularity of a configuration can be strengthened to become 
a symmetricity.

\begin{lemma}
\label{hom:lem:symEqReg}
Let $P$ be any configuration with $reg(P)=m>1$, and let $c$ be its center of regularity.
There exists a configuration $P^\prime$ that can be obtained from $P$ by
making robots move along their radius such that $sym(P^\prime)=m$.
Moreover, $c$ is the center of symmetricity of $P^\prime$.
\end{lemma}

\begin{proof}
Let $d$ be the greatest distance between any point of $P$ and $c$.
Note that since $m>1$, no robot is located in $c$.
We construct a configuration $P^\prime$ by making each point $r$ in $P$ move along its radius towards the point $(d,0)$ 
(expressed in its local polar coordinate system).
In  configuration $P^\prime$, all the points are at the same distance from $c$, hence the symmetricity in this case depends only on angles
which makes it equal to regularity.
Indeed, the equivalence class of each robot $r$ is formed by the set of robots $\{r, S^{n/m}(r,c), \ldots, S^{(m-1)n/m}(r,c)\}$.
Clearly, $|[r]|=m$. Since $r$ is arbitrary, all the equivalence classes have a cardinality of $m$.
Hence, $sym(P^\prime)=m$ and $c$ is the center of symmetricity of $P^\prime$.
\end{proof}

The following lemma states that the center of symmetricity of any configuration $P$ if any is also its Weber point.
The same claim was made in \cite{ACP05} in the cases when $sym(P)=n$ (equiangular) and $sym(P)=n/2$ (biangular).
Our proof uses the same reasoning as theirs.

\begin{lemma}
\label{hom:lem:weberSym}
Let $P$ be a configuration such that $sym(P)>1$ and let $c$ be its center of symmetricity.
It holds that $c$ is the Weber point for $P$.
\end{lemma}

\begin{proof}
Let $m=sym(P)$.
Assume towards contradiction that the Weber point of $P$ is $c^\prime \neq c$.
Since $P$ is $m$-symmetric, it holds that all there $m-1$ other points that are rotational symmetric 
to $c^\prime$ around $c$ with angles of $\frac{2 \cdot \Pi}{m}, \ldots \frac{2\cdot (m-1)\cdot \Pi}{m}$.
By symmetricity, all these points are also Weber points.
But since $m>1$, this means that the Weber point is not unique.
But Weber points are unique when not all the points are located on the same line.
Hence $c$, the only point that is \emph{"unique"} in $P$, is the Weber point.
\end{proof}

\begin{lemma}
\label{hom:lem:weberReg}
Let $P$ be a configuration such that $reg(P)=m>1$ and let $c$ be its center of regularity.
It holds that $c$ is the Weber point for $P$.
\end{lemma}

\begin{proof}
Our argument is similar to \cite{ACP05}.
According to Lemma \ref{hom:lem:symEqReg}, there exists a configuration $P^\prime$ that can be obtained from $P$ by
making robots move along their radius. Moreover, $sym(P^\prime)=m$ and $c$ is its center of symmetricity.
According to Lemma \ref{hom:lem:weberSym}, $c$ is the Weber point of $P^\prime$.
But it is known in geometry that a Weber point remains invariant under straight movements of any
of the points towards or away from it.
Hence, since $c$ is the Weber point of $P^\prime$, it must be
the Weber point of $P$ also.
This proves the lemma.
\end{proof}

The next corollary follows directly from Lemma \ref{hom:lem:weberReg} and the fact
that Weber points are unique.

\begin{corollary} [Unicity of $c$]
The center of regularity is unique if it exists.
\end{corollary}
\section{Detection of Regular Configurations}
\label{hom:sec:regular}

In this section we show how to identify geometric configurations that are regular.
We present two algorithms 
that detects whether a configuration $P$ of $n$ points given in input is $m$-regular for some $m>1$, and if so, they output
its center of regularity.
The first one detects the regularity only if $m$ is even, it is very simple and runs in $0(n~logn)$ time.
The second one can detect any regular configuration, provided that $m\geq 3$.
It is a little more involved and runs in $0(n^4~log n)$ time.
We assume in this section that $P$ is not a configuration in which all the points are collinear.

\subsection{Preliminaries}
In this section, we state some technical lemmas that help us in the presentation and the proofs.

\begin{lemma}
\label{hom:lem:2piSurm}
Let $P$ be a regular configuration of $n$ distinct points and let $c$ be its center of regularity.
Let $m=reg(P)$.
The following property holds:

$$
\forall r \in P: 
(\sphericalangle(r, c, S^{n/m}(r, \rightturn), \rightturn) = 2\Pi/m)
\wedge
(\sphericalangle(r, c, S^{n/m}(r, \leftturn), \leftturn) = 2\Pi/m)
$$
\end{lemma}


\begin{proof}
Fix $r \in P$.
Since $P$ is $m$-regular, then $m$ is necessarily a divisor of $n$.
Let $x=n/m$. We divide the proof into two parts:

\begin{enumerate}

\item \textbf{Part 1:} $\sphericalangle(r, c, S^{x}(r, \rightturn), \rightturn) = 2\Pi/m$

Let $S=SA(r,c,\rightturn)$.
Assume that $S=\alpha_1 \ldots \alpha_n$.
Clearly we have that $$\alpha_1 + \alpha_2 + \ldots + \alpha_n=2\Pi$$
Since $n=m \cdot x$, the above sum can be rewritten as follows:
$$
(\alpha_1 + \ldots + \alpha_x )+ (\alpha_{x+1}+ \ldots + \alpha_{2x}) + \ldots + (\alpha_{(m-1)x+1}+ \ldots+\alpha_{mx}) = 2\Pi
$$
But $S$ is $m$-regular. Hence, 
$\forall i \in [1,n]: \alpha_i = \alpha_{(i+x) mod (n+1)}$. 
This implies that:
$$
m \cdot (\alpha_1 + \ldots + \alpha_x ) = 2\Pi
$$
But $\sphericalangle(r, c, S^x(r, \rightturn), \rightturn) = (\alpha_1 + \ldots + \alpha_x)$.
Hence $\sphericalangle(r, c, S^x(r, \rightturn), \rightturn)= 2\Pi /m$. 
This finishes the first part of the proof.
\\

\item \textbf{Part 2:} $\sphericalangle(r, c, S^{x}(r, \leftturn), \leftturn) = 2\Pi/m $

Let $r^\prime= S^{x}(r, \leftturn)$ and notice that $S^{x}(r^\prime, \rightturn)=q$.
It holds that

$$
\sphericalangle(r, c, S^x(r, \leftturn), \leftturn)= \sphericalangle(r, c, r^\prime, \leftturn)= \sphericalangle(r^\prime, c, r, \rightturn)
$$
But since $S^{x}(r^\prime, \rightturn)=r$, we have:
$\sphericalangle(r^\prime, c, r, \rightturn)=\sphericalangle(r^\prime, c, S^x(r^\prime, \rightturn), \rightturn)=2\Pi /m$ (as shown in Part 1).
Hence $\sphericalangle(r, c, S^x(r, \leftturn), \leftturn)=2\Pi /m$.

\end{enumerate}


\end{proof}


The following lemma proves that when a configuration $P$ is $m$-regular with $m$ even, then
for each point in $r \in P$ there exists a corresponding point ($S^{n/2}(r)$) that lies
on the line that passes through $r$ and $c$ with $c$ being the center of regularity of $P$.
\begin{lemma}
\label{hom:lem:pisur2}

Let $P$ be a regular configuration of $n$ distinct points and let $c$ be its center of regularity.
The following property holds:

$$
(reg(P) \text{ is even}) 
\Rightarrow
(\forall r \in P: \sphericalangle(r, c, S^{n/2}(r)) = \Pi)
$$
\end{lemma}

\begin{proof}
Let $m=reg(P)$. By assumption $m$ is even, hence there exists some $m^\prime \in \mathbb{N^{+}}$ such that $m=2 m^\prime$.
Then, it holds that $n/2=m^\prime (n/m)$.
This implies that 
$$
\sphericalangle(r, c, S^{n/2}(r))=\sphericalangle(r, c, S^{n/m}(r))+\sphericalangle(S^{n/m}(r), c, S^{2n/m}(r))
+ \ldots + \sphericalangle(S^{(m^\prime-1)n/m}(r), c, S^{n/2}(r))
$$
But since $P$ is $m$-regular, we have according to Lemma \ref{hom:lem:2piSurm} that:

$$
\sphericalangle(r, c, S^{n/2}(r))=2\Pi/m + \ldots + 2\Pi/m=m^\prime(2\Pi/m)
$$
By replacing $m=2m^\prime$ we get: $\sphericalangle(r, c, S^{n/2}(r))=\Pi /2$ which proves the claim.
\end{proof}


\subsection{Detection of Even Regularity}

In this section we present our algorithm for detecting regular configurations when the regularity is even.
Note that since regularity is a divisor of $n$, it holds that $n$ must be also even in this case.
It is inspired from Algorithm 2 in \cite{ACP05} and runs in $0(n~log n)$ steps.
It is based on the notion of median lines taken from \cite{ACP05}.

\begin{definition}[$median(r)$]
Given a configuration $P$ of $n$ robots, $n$ even, given $r$ any robot in $P$, 
the median line of $r$, denoted $median(r)$, is the line that passes through $r$ and
some other robot $r^\prime$ in $P$ and divides the set of points into two subsets of equal cardinality 
$\frac{n}{2}-1$. Note that in our case, contrary to \cite{ACP05}, we allow other robots 
than $r$ and $r^\prime$ to lie in the median. This is explained by the fact that, contrary to them, we allow 
angles equal to $0$.
\end{definition}

\begin{lemma}
Let $P$ be a regular configuration of $n$ distinct points with $reg(P)=m$ even.
Let $c$ be the center of regularity.
The following property holds:
$$
\forall r \in P: median(r) \text{ passes through $S^{n/2}(r)$ and $c$}
$$
\end{lemma}

\begin{proof}
According to Lemma \ref{hom:lem:pisur2}, $\sphericalangle(r, c, S^{n/2}(r))= \pi$. Hence the lemma follows.

\end{proof}

As a corollary we have that the center of regularity lies in the intersection of all medians.

\begin{lemma}
\label{hom:lem:intersectionMedians}
Let $P$ be a regular configuration of $n$ distinct points with $reg(P)=m$ even.
Let $c$ be the center of regularity. It holds that
$$\forall r \in P: median(r) \cap \{c\}=\{c\}$$
\end{lemma}

\begin{theorem}
\label{hom:thm:testRegularityEven}
Given $P$ a configuration of $n$ distinct points with $n$ even. 
There exists an algorithm running in $0(n~log n)$ steps that detects if $P$ is $m$-regular with $m$ even,
and if so, it outputs $m$ and the center of regularity.
\end{theorem}

\begin{proof}
The algorithm is described formally in Figure \ref{hom:alg:evenRegularity}.
It first computes the convex hull $CH$ of the configuration which takes $0(n~log n)$ operations.
Then it chooses an arbitrary point $r$ in $CH$ and computes its $median(r)$.
After that, it picks another point $r^\prime$ that belongs to $CH$ but not to $median(r)$.
The point $r^\prime$ is well defined since we assume that the points of $P$ are not all collinear.
This choice of $r^\prime$ guarantees  that $median(r) \neq median(r^\prime)$.
If the configuration is regular, its center of regularity necessarily lies in the intersection $c$
of the two medians (Lemma \ref{hom:lem:intersectionMedians}).
Then, it can be checked in $0(n~log n)$ steps whether $c$ is a center of regularity (Theorem \ref{hom:thm:algoDetectRegularity}).

\begin{figure}[htb]
\centering{ \fbox{
\begin{minipage}[t]{150mm}
\scriptsize 
\renewcommand{\baselinestretch}{2.5} \resetline{}
\begin{tabbing}
aaaaa\=aa\=aaa\=aaa\=aaaaa\=aaaaa\=aaaaaaaaaaaaaa\=aaaaa\=\kill 

\line{T1} \> $CH \leftarrow \textsc{Convex Hull}(P)$ \\

\line{T2} \> $r \leftarrow \text{ an arbitrary point on $CH$}$ \\
\line{T3} \> $r^\prime \leftarrow \text{a point on $CH$ but not on $median(r)$ }$\\
\line{T4} \> \textbf{If} ($(median(r) \cap median(r^\prime)) \neq \bot$) \\
\line{T5} \>\> $c \leftarrow median(r) \cap median(r^\prime)$ \\
\line{T6} \>\> $SA \leftarrow \text{ string of angles of $P$ around $c$ }$\\
\line{T7} \>\> $m \leftarrow per(SA)$ \\
\line{T8} \>\> \textbf{If} $(m>1)$ \textsc{return} (\textsc{Regularity $m$ with Center $c$}) \\
\line{T9} \> \textbf{Endif} \\
\line{T10} \> \textsl{return} (\textsc{Not Regular}) \\

\end{tabbing}
\normalsize
\end{minipage}
}
\caption{Detection of Even Regularity}
\label{hom:alg:evenRegularity}
}
\end{figure}

\end{proof}

\subsection{Detection of Odd Regularity}

\begin{definition}[$\alpha$-Circle]
\label{hom:def:thales}

Given two distinct points $x$ and $y$ and an angle $0 < \alpha < \Pi$, 
we say that the circle $C_{xy}$ is a $\alpha$-circle for $p$ and $q$ if $x, y \in C_{pq}$ and there \emph{exists}
a point $p \in C_{xy}$ such that $\sphericalangle(x, p, y)=\alpha$.
\end{definition}
In the following we present three known properties \cite{ACP05} about $\alpha$-circles.

\begin{property}
\label{hom:prop:thales}
it holds that $\sphericalangle(x, p^\prime, y)=\alpha$
for every point $p^\prime \in C_{pq}$ on the same arc as $p$.
\end{property}

\begin{property}
\label{hom:prop:thales90}
If $\alpha=\Pi/2$, the $\alpha$-circle is \emph{unique} and is called the Thales circle.
\end{property}

\begin{property}
\label{hom:prop:2thalesNot90}
If $\alpha \neq \Pi/2$, there are \emph{exactly} two $\alpha$-circles. 
We denote them in the following by $C_{xy}$ and $C^\prime_{xy}$.
$\alpha = \Pi/2$ can be seen as a special case in which $C_{xy}=C^\prime_{xy}$.
\end{property}

\begin{definition}[$C_{xy} \cap C_{yz}$]
Given $C_{xy}$ and $C_{yz}$, we define their intersection, denoted $C_{xy} \cap C_{yz}$,
as the point $p$ such that $(p \neq y) \wedge (p \in C_{xy}) \wedge (p \in C_{yz})$.
If $p$ does not exists we write $C_{xy} \cap C_{yz}=\emptyset$.
\end{definition}

\begin{lemma}
\label{hom:lem:cIntersectionCircles}
Let $m\geq 3$. Let $P$ be an $m$-regular configuration of $n$ distinct points with center $c \not\in P$.
Let $x$ be any point in $P$, 
and let us denote by $y$ and $z$ 
the points $S^{n/m}(x,c,\rightturn)$ and $S^{n/m}(x,c,\leftturn)$ respectively.
Let $C_{xy}, C_{xy}^\prime, C_{xz}, C_{xz}^\prime$ be the $2\Pi/m$-circles for the appropriate points.
It holds that 
$c \in  (C_{xy} \cap C_{xz}) \cup (C_{xy}^\prime \cap C_{xz}) \cup (C_{xy} \cap C_{xz}^\prime) \cup (C_{xy}^\prime \cap C_{xz}^\prime)$.

\end{lemma}

\begin{proof}
\label{hom:lem:cIntersectionCircles-bis}

According to Theorem \ref{hom:lem:2piSurm}, since $P$ is $m$-regular, then
$\sphericalangle(x, c, y, \rightturn) = 2\Pi/m$ where $c$ is the center of regularity.
Hence, either $c \in C_{xy}$ or $c \in C_{xy}^\prime$ (According to Definition \ref{hom:def:thales} and Property \ref{hom:prop:2thalesNot90}).
Using the same argument, we can also show that either $c \in C_{xz}$ or $c \in C_{xz}^\prime$.
Hence: 
$$((c \in C_{xy})  \vee (c \in C_{xy}^\prime)) \wedge ((c \in C_{xz}) \vee (c \in C_{xz}^\prime))$$

Distributing $\wedge$ over $\vee$ gives us:
$$((c \in C_{xy}) \wedge (c \in C_{xz}))
\vee ((c \in C_{xy}) \wedge (c \in C_{xz}^\prime))
\vee ((c \in C_{xy}^\prime) \wedge (c \in C_{xz}))
\vee ((c \in C_{xy}^\prime) \wedge (c \in C_{xz}^\prime))
$$

This proves the lemma.
\end{proof}

\begin{lemma}
Given $3 \leq m\leq n$, given $P$ a configuration of $n$ distinct points. 
Let $x$ be any point in $P$
The following property holds:

$$(P \text{ is $m$-regular with center } c)$$
$$ \Leftrightarrow $$
$$
(\exists y,z \in P: (x \neq y \neq z) 
\wedge 
(c
\in  
(C_{xy} \cap C_{xz}) \cup (C_{xy}^\prime \cap C_{xz}) \cup (C_{xy} \cap C_{xz}^\prime) \cup (C_{xy}^\prime \cap C_{xz}^\prime)
))
$$
Where $C_{xy},C_{xz},C_{xy}^\prime,C_{xz}^\prime$ are $2\Pi/m$-circles of the corresponding points.
\end{lemma}

\begin{proof}
The lemma follows from Lemma \ref{hom:lem:cIntersectionCircles} by setting 
$y=S^{n/m}(x,c,\rightturn)$ and $z=S^{n/m}(x,c,\leftturn)$.
\end{proof}

\begin{theorem}
\label{hom:thm:testRegularityOdd}
Given $3 \leq m\leq n$, given $P$ a configuration of $n$ distinct points. 
There exists an algorithm running in $0(n^3~log n)$ that detects if $P$ is $m$-regular,
and if so, it outputs the center of regularity.
\end{theorem}

\begin{proof}
The algorithm is the following. We fix any robot $x \in P$.
Then, for every $y \in P \setminus \{x\}$, for every $z \in P \setminus \{x,y\}$, 
for every $c \in (C_{xy} \cap C_{xz}) \cup (C_{xy}^\prime \cap C_{xz}) \cup (C_{xy} \cap C_{xz}^\prime) \cup (C_{xy}^\prime \cap C_{xz}^\prime)$,
we test if $c$ is a center of regularity (Theorem \ref{hom:thm:algoDetectRegularity}, $O(n~log n)$ time).
Lemma \ref{hom:lem:cIntersectionCircles-bis} guarantees that if $P$ is $m$-regular, the test will be conclusive for at least one pair $(y,z)$ of robots.
The whole algorithm executes in $O(n^3~log n)$:
we browse all the possible pairs $(y, z)$, and for each pair we generate up to four candidates for the center of regularity, hence we have $O(n^2)$ 
candidates. Then, $O(n~log n)$ time is needed to test each candidate.
Note that our algorithm follows the same patterns as those presented in \cite{ACP05}: generating a restricted set of candidates (points) and 
testing whether each of them is a center of regularity.
\end{proof}

\begin{theorem}
\label{hom:thm:testRegularityOdd2}
Given $P$ a configuration of $n$ distinct points. 
There exists an algorithm running in $0(n^4~log n)$ steps that detects if $P$ is $m$-regular with $m\geq 3$,
and if so, it outputs $m$ and the center of regularity.
\end{theorem}

\begin{proof}
It suffices to generates all the divisors $m$ of $n$ that are greater than 2.
Then, for each $m$, we test if $P$ is $m$-regular using the algorithm of Theorem \ref{hom:thm:testRegularityOdd}).
When the test is conclusive, this algorithm return the center of regularity $c$, so we can output 
\textsc{Regularity $m$, Center $c$}.
If test was inconclusive for every generated $m$, we output \textsc{Not Regular}
\end{proof}

\begin{theorem}
\label{hom:thm:testRegularityOdd3}
Given $P$ a configuration of $n$ distinct points. 
There exists an algorithm running in $0(n^4~log n)$ steps that detects if $P$ is $m$-regular with $m\geq 2$,
and if so, it outputs $m$ and the center of regularity.
\end{theorem}

\begin{proof}
We combine the algorithms of Theorems \ref{hom:thm:testRegularityEven} and \ref{hom:thm:testRegularityOdd2}.
First, we test if $P$ is $m$-regular for some even $m$ using the algorithm of Theorem \ref{hom:thm:testRegularityEven}  ($0(n~log n)$).
If so, we output $m$ and the center of regularity $c$ which are provided by the called algorithm.
Otherwise, we test odd regularity using the algorithm of Theorem \ref{hom:thm:testRegularityOdd2} ($0(n^4~log n)$)
but
by restricting the analysis to only the \emph{odd} divisors of $n$ (the even divisors were already tested).

\end{proof}
\section{Formation of a Series of Geometric Patterns: Lower Bound}
\label{hom:sec:formationImpossibility}

In this section we prove a necessary condition
that geometric series have to satisfy in order to be formable.
The condition relates three parameters: the number of robots in the system $n$, its homonymy $h$
and the symmetricity of the patterns to form. It is stated in Lemma \ref{hom:lem:necessaryFormationSeries}.

\begin{property}\cite{DFSY10} 
\label{hom:prop:symDividesN}
For any configuration $P$ of $n$ distinct robots, $sym(P)$ divides $n$.
\end{property}

\begin{lemma}
\label{hom:lem:relationHandM}
Let $P$ be a configuration of $n$ distinct robots with symmetricity $s$, \emph{i.e.} $sym(P)=s$.
For any divisor $d$ of $s$, if $h \leq \frac{n}{d}$, then for any pattern formation algorithm,
there exists an execution where all subsequent configurations $P^\prime$ satisfy
$sym(P^\prime)= k\cdot d$, $k >1$.
\end{lemma}

\begin{proof}
The lemma holds trivially if $s=1$, hence we assume in the following that $s>1$.
According to Lemma \ref{hom:lem:setConcentricCircles}, there exists
a partition of $P$ into $x=n/s$ subsets $S_1, \ldots, S_x$ such that
the $s$ robots in each $S_i$ occupy the vertices of a regular convex polygon of 
$s$ sides whose center is $c$.

Now, partition each set $S_i$ into $\frac{s}{d}$ subsets $T_{i1}, \ldots, T_{(i\frac{s}{d})}$ with
$|T_{ij}|=d$ for each $i \in \{1, \ldots, x\}, j \in \{1, \ldots, \frac{s}{d}\}$. 
Each subset $T_{ij}$ is chosen in such a way that the $d$ robots belonging to it
are located in the vertices of a regular 
convex polygon of $d$ sides with center $c$. 
This choice is possible because $d$ is a divisor of $s=|S_i|$.
For example, let $r_1, \ldots, r_s$ be the robots of $S_i$ ordered according to some polar ordering around $c$.
$T_{i1}$ is the set of robots $\{r_1, r_{(\frac{s}{d}+1)}, r_{(\frac{2s}{d}+1)}, \ldots, r_{(\frac{(d-1)s}{d}+1)}\}$. 
Clearly, this subset defines 
a regular polygon of $d$ sides with center $c$.

There are total of  $\frac{s\cdot x}{d}= \frac{n}{d}$ subsets $T_{ij}$. So we have also 
a total of $\frac{n}{d}$ concentric regular polygons of $d$ sides.
What is important to notice now is that the robots in each $T_{ij}$ have the same view.

Since $h \leq \frac{n}{d}$, there exists a set of labels $|\mathcal{L}|=h$, and a labeling of robots in $P$ such that 
(1) The same label is assigned to the robots that belong to the same subset $T_{ij}$.
(2) For each label $l \in \mathcal{L}$, there exists a robot $r_i \in P$ such that $l$ is the label of $r_i$.

Since we assume that algorithms are deterministic, the actions taken by robots at each activation depend solely on 
they observed view and their identity (label).
Hence, two robots having the same view and the same label will take the same actions if they are activated simultaneously.
Therefore, the adversary can guarantee that the network will always 
have a symmetricity $\geq d$ 
by activating each time the robots
that belong to the same $T_{ij}$ together. 
This way, we are guaranteed to have all the subsequent configurations
that consists of a set of $\frac{n}{d}$ concentric 
regular polygons of $d$ sides.
That is, all subsequent configuration have a symmetricity that is a multiple of $d$.
This proves the lemma.
\end{proof}

The following lemma states a necessary condition for a geometric figure 
$P_j$ to be formable starting from $P_i$.

\begin{lemma}
\label{hom:lem:necessaryFormationPattern}
If the current configuration $P_i$ has symmetricity $sym(P_i)=s$,
the configuration $P_j$ with symmetricity $sym(P_j)=s^\prime$
is formable only if  (1) $size(P_j)=size(P_i)$ and (2)
$h > \frac{n}{\textsc{sncd}(s, s^\prime)}$ where \textsc{sncd} (read smallest non common divisor) is equal to
the smallest $x$ that divides $s$ but not $s^\prime$ if any, $n+1$ otherwise.
\end{lemma}

\begin{proof}
Assume towards contradiction that 
(1) $h \leq \frac{n}{\textsc{sncd}(s, s^\prime)}$ and (2) $P_j$ is formable.
Note that since $h \geq 1$, (1) implies that $\textsc{sncd}(s,s^\prime) \neq n+1$. 
Otherwise we would have $h \leq 0$, contradiction.
By definition, $\textsc{sncd}(s,s^\prime) \neq n+1$ implies that
 $\textsc{sncd}(s,s^\prime) =t$ divides $s$ but not $s^\prime$.

According to Lemma \ref{hom:lem:relationHandM}, for any algorithm,
there exists an execution starting from $P_i$ where all subsequent configurations $P^\prime$ satisfy
$sym(P^\prime)= k\cdot t$, $k >1$. 
But $sym(P_j)=s^\prime$ is not multiple of $t$, hence $P_j$ is never reached in this execution.
This means that $P_j$ is not formable starting from $P_i$, which contradicts (2).
Hence, the lemma is proved.
\end{proof}

Now, we are ready to state the necessary condition for formation of geometric series.
It relates the symmetricity of its constituent patterns and the homonymy of the system.

\begin{lemma}
\label{hom:lem:necessaryFormationSeries}
A cyclic series of distinct patterns 
$\langle P_1, P_2, ..., P_m \rangle^{\infty}$ each of size $n$ is formable only if 
$$\forall i \in \{1, \ldots, m\}: 
h > \frac{n}{\textsc{sncd}(sym(P_i), sym(P_{i+1}))}$$ 
\end{lemma}

\begin{proof}
Follows from Lemma \ref{hom:lem:necessaryFormationSeries}.
\end{proof}

\section{Formation of Series of Patterns: Upper Bound}
\label{hom:sec:formationPossibility}
In this section, we present an algorithm that allows robot to form a series of patterns,
provided that some conditions about homonymy and symmetricity are satisfied. 
The result is stated in the following theorem:

\begin{theorem}
\label{hom:thm:formationPossibility}

A cyclic series of distinct patterns 
$\langle P_1, P_2, ..., P_m \rangle^{\infty}$ each of size $n$ is formable if and only if 
$$\forall i \in \{1, \ldots, m\}: 
h > \frac{n}{\textsc{sncd}(sym(P_i), sym(P_{i+1}))}$$ 
\end{theorem}

The only if part was proved in Section \ref{hom:sec:formationImpossibility}. The remaining of the 
current section is devoted to the proof of the if part of the theorem.

\subsection{Intermediate Configurations}

During the formation of a pattern $P_i$ (starting from $P_{i-1}$), the network 
may go through several intermediate configurations. We define in the following four
classes of intermediate patterns $\mathcal{A},\mathcal{B}, \mathcal{C}$ and $\mathcal{D}$. 
Each one of them encapsulate some information that allows
robots to unambiguously determine which pattern the network is about to form.
This information is provided by a function, $\textsc{Stretch}$, which we define separately for each 
intermediate pattern.

\begin{definition}[Configuration of type $\mathcal{A}$]
A configuration $P$ of $n$ points is of type $\mathcal{A}$ (called $BCC$ in \cite{DFSY10})  
if the two following conditions are satisfied (refer to Figure \ref{hom:fig:patternTypeA}):
\begin{enumerate}
\item there exists a point $x \in P$ such that
the diameter of $SEC_1=SEC(P)$ is a least ten times the diameter of $SEC_2=SEC(P \setminus \{x\})$.
\item $SEC_1$ and $SEC_2$ intersect at exactly one point called the \emph{base-point} (BP).
\end{enumerate}
The point $x$ is called the \emph{pivot} whereas the point on $SEC_2$ directly opposite $BP$ is called
the \emph{frontier point} $(FP)$.

$\textsc{Stretch}(P)$ is equal to $\lfloor\frac{rad(SEC_1)}{(h+1) \cdot rad(SEC_2)}\rfloor$.
\end{definition}

\begin{figure}
\label{hom:fig:patternTypeA}
\begin{center}

\epsfig{figure=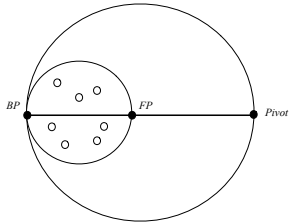,width=4cm}\\
\caption{Figure of type $\mathcal{A}$.}
\end{center}
\end{figure}


\begin{definition}[Configuration of type $\mathcal{B}(m)$]
A configuration $P$ of $n$ distinct robots is of type  $\mathcal{B}(m)$ (adapted from $SCC[m]$ in \cite{DFSY10}) with $1<m<n$, 
if the following conditions are satisfied:
\begin{enumerate}
\item $SEC_1=SEC(P)$ has exactly $m$ points on its circumference which form a regular convex polygon with $m$ sides.
\item Let $SEC_2$ be the SEC of the robots that are not on the $SEC_1$, \emph{i.e.}
$SEC_2=SEC(P\setminus \{r \in P ~|~ r \text{ is on }SEC_1\})$.
$SEC_1$ and $SEC_2$ are concentric such that $rad(SEC_1) > 10 \cdot rad(SEC_2)$.
\end{enumerate}

$\textsc{Stretch}(P)$ is equal to $\frac{1}{2} \cdot \lfloor\frac{rad(SEC_1)}{(h+1) \cdot rad(SEC_2)}\rfloor$.
It can be easily checked that a given configuration cannot be of both types $\mathcal{A}$ and
$\mathcal{B}$ (their respective $SEC_2$ do not intersect).
\end{definition}


\begin{definition}[Configuration of type $\mathcal{C}(m)$]
A configuration $P$ of $n$ distinct robots is of type $\mathcal{C}(m)$ with $1<m<n$ if 1) it is not of type $\mathcal{B}$
and 2) it is $m$-symmetric (ref. Definition \ref{hom:def:symmetricity}).

$\textsc{Stretch}(P)$ is computed as follows.
Since $P$ is $m$-symmetric with $m>1$, 
there exists a partition of $P$ into $x=n/m$ subsets $S_1, \ldots, S_x$ such that
each of them is a convex regular polygon of $m$ sides with center $c$ (Lemma \ref{hom:lem:setConcentricCircles}).
This means that each $S_i$ defines a circle with center $c$.
Assume w.l.o.g that $\forall i \in \{1 \ldots x-1\}: rad(S_i) \leq rad(S_{i+1})$.
$\textsc{Stretch}(P) = Max\{\lfloor\frac{rad(S_{i+1})}{(h+1) \cdot rad(S_i)}\rfloor ~|~ i \in \{1 \ldots x-1\}\}$.

Clearly, a configuration of type $\mathcal{C}$ cannot be of type $\mathcal{A}$ since the former is symmetric while the latter is not.
\end{definition}

\begin{definition}[Configuration of type $\mathcal{D}(m)$]
A configuration $P$ of $n$ distinct robots is of type $\mathcal{D}(m)$ 
with $1<m<n$ if (1) Points are not all on the same line, (2) $P$ is not of type $\mathcal{B}$
and (3) $P$ is $m$-regular but not symmetric ($Sym(P)=1$).

$\textsc{Stretch}(P)$ is computed as follows. Let $c$ be the center of regularity.
$c$ can be computed in polynomial time using the algorithm of Section \ref{hom:sec:regular} (Theorem \ref{hom:thm:testRegularityOdd3}).
Then for each $r_i \in P$, we compute $|c, r_i|$, its distance from $c$.
Assume w.l.o.g. that $\forall i \in \{1 \ldots x-1\}: |c, r_i| \leq |c, r_{i+1}|$.
$\textsc{Stretch}(P) = Max\{\lfloor\frac{|c, r_{i+1}|}{(h+1) \cdot |c, r_i|}\rfloor ~|~ i \in \{1 \ldots x-1\}\}$.
Note that the stretch of configurations of type $\mathcal{C}(m)$ is a particular case
of that of type $\mathcal{D}(m)$, but since the former configurations are symmetric, we can 
compute their stretch without resorting to the computation of the center of regularity.
\end{definition}

\paragraph{\textbf{Decoding the stretch}}
Let $F$ be a one-to-one function \cite{DFSY10} that maps each pattern $P_i$ to a real number $t_i=F(P_i)$.
If there is a pattern $P_i$ that is of type $\mathcal{A}, \mathcal{B}, \mathcal{C}$ or $\mathcal{D}$,
we exclude the value $\textsc{Stretch}(P_i)$ from the domain of $F$.
To simplify the proofs, we assume that $F(P_i)>10$ for any $P_i$.
When robots are about to form the pattern $P_i$, they use intermediate configurations with stretch $t_i$.
By computing the stretch, robots can unambiguously identify which configuration they are about to form $(F^{-1}(t_i))$.


\subsection{Transitions between Configurations}
In this section we describe some algorithms that describe some transformations between 
patterns.

\begin{lemma}
\label{hom:lem:D2A}
Starting from any configuration of type $\mathcal{D}$ with stretch $t$, there exists an algorithm
that builds a configuration of type $\mathcal{A}$ with the same stretch.
\end{lemma}

\begin{proof}
Let $P$ be the initial configuration. Since $P$ is of type $\mathcal{D}$, it holds that $sym(P)=1$.
Hence we can elect a leader, \emph{i.e.} a robot whose view is unique, let it be $r$.
The algorithm work by making $r$ move towards the pivot position of the target $\mathcal{A}$ configuration
with appropriate stretch.
This is done as follows.
First, $r$ computes the stretch $t=\textsc{Stretch}(P)$.
Then it consider the set of points $Q=P \setminus \{r\}$ and compute their $SEC$ and its center $c$.
Here, we distinguish between two cases:
\begin{enumerate}
\item 
There exists two robots $r_2$ and $r_3$ of $SEC$ that are directly opposite of each others. 
In this case, $r_1$ jumps to the pivot location that makes $r_2$ and $r_3$ occupy the frontier-point 
and the base-point respectively. 
Formally, $r_1$ jumps to a point $x$ located at the line $(r_3, r_2)$
such that $r_2$ lies between $r_3$ and $x$ 
with $|c, x|= (t \cdot (h+1)) \cdot |c, r_2|$.
Let $P^\prime$ the obtained configuration. 
Clearly, $P^\prime$ is of type $\mathcal{A}$ 
and $\textsc{Stretch}(P^\prime)=t$.

\item Otherwise, $r_1$ chooses any point $r_2$ on $SEC$ and jumps to the pivot position with 
appropriate stretch such that $r_2$ becomes the base point. Then, any robot distinct from $r_1$ and $r_2$ 
will jump to the frontier point.
\end{enumerate}
In both cases, we obtain a configuration of type $\mathcal{A}$ with the required stretch.
\end{proof}

\begin{lemma}
\label{hom:lem:P2A}
Starting from any configuration of type $P_i$ with stretch $sym(P_i)=1$, there exists an algorithm
that builds a configuration of type $\mathcal{A}$ with stretch $F(P_{i+1})$.
\end{lemma}
\begin{proof}
The proof is similar to that of Lemma \ref{hom:lem:D2A}, by replacing $t$ with $F(P_{i+1})$. 
\end{proof}

The following two lemmas are from \cite{DFSY10}.

\begin{lemma}
\label{hom:lem:A2anyPattern}
Starting from any configuration of type $A$, it is possible to form any single pattern.
\end{lemma}

\begin{lemma}
\label{hom:lem:B2somePatterns}
Starting from any configuration of type $B(m)$, it is possible to form any single pattern $P$ such that $sym(P)=k \cdot m, k>1$.
\end{lemma}

\begin{lemma}
\label{hom:lem:grandAlgoChiantEtDeprimant}
Consider a robot network of $n$ robots in configuration $P_i$. Let $sym(P_i)=m>1$ and $sym(P_{i+1})=m^\prime$.
If $h> \frac{n}{\textsc{sncd}(m, m^\prime)}$, there exists an algorithm that brings the network to a configuration $Q$ such that 
either $Q$ is of type $\mathcal{B}(x), x < \textsc{sncd}(m, m^\prime)$ or of type $\mathcal{A}$ both with stretch $F(P_{i+1})$.
\end{lemma}

The remaining of this section is devoted to the proof of this lemma.

\begin{figure}[htb]
\centering{ \fbox{
\begin{minipage}[t]{150mm}
\scriptsize 
\renewcommand{\baselinestretch}{2.5} \resetline{}
\begin{tabbing}
aaaaa\=aa\=aaa\=aaa\=aaaaa\=aaaaa\=aaaaaaaaaaaaaa\=aaaaa\=\kill 

\textbf{Function:} \> \\
$targetSym(k):$ The target symmetricity when robots try to form $P_{k+1}$\\
It is equal to the greatest divisor of $sym(P_{k+1})$ that is $< \textsc{SNCD}(sym(P_k), sym(P_{k+1}))$\\
~\\

\textbf{Actions:}\\
\line{A1} \> $P \leftarrow \text{ Observed Configuration}$ \\
\line{A2} \> $SEC \leftarrow \textsc{Smallest Enclosing Circle}(P)$\\
\line{A3} \> $c \leftarrow \textsc{Center}(SEC)$ \\
\line{A4} \> $sym \leftarrow sym(P)$ \\
\line{A5} \> \textbf{If} $(P=P_i)$ \textbf{then} $t \leftarrow F(P_{i+1})$ \\
\line{A6} \> \textbf{else} $t \leftarrow stretch(P)$ \textbf{endif}\\
\line{A7}\> $id \leftarrow \text{My Identifier}$ \\
\line{A8} \> $rad \leftarrow \textsc{Radius}(SEC)$ \\
\line{A9} \> \textbf{if} $(sym > 1) \wedge ((sym > targetSym(k)) \vee (P \text{ is not of type } \mathcal{B}(sym)))$\\
\line{A10} \>\> $d \leftarrow Min(|r_i, c|; r_i \in P)$ \\
\line{A11} \>\> $S \leftarrow \{r_i \in P~|~ |r_i,c|=d\}$ \\
\line{A12} \>\> $minView \leftarrow Min(V(r_i); r_i \in S)$ \\ 
\line{A13} \>\> $Elected \leftarrow \{r_i \in S~|~ V(r_i)=minView\}$ \\
\line{A14} \>\> \textbf{if} $(r_i \in Elected)$ \\
\line{A15} \>\>\> return ($(t \cdot (h+1)+id-1) \cdot rad$, 0) \\
\line{A16} \>\> \textbf{else} \\
\line{A17} \>\>\> return \text{My Position} \\
\line{A18} \>\>\textbf{endif} \\
\line{A19} \> \textbf{endif} \\

\end{tabbing}
\normalsize
\end{minipage}
}
\caption{Symmetry Breaking}
\label{hom:alg:formation1}
}
\end{figure}

The transformation algorithm is given in Figure \ref{hom:alg:formation1}.
It is executed by robots during their Compute phases until one of the two desired configurations is obtained.
Its description is based on the polar coordinate system of Section \ref{hom:subsec:polarCoordinate}.
The principal idea is to use identifiers of robots in order to break 
the symmetry of configurations. It does so by making robots choose their destination according to their identity.
This way, if two robots with similar views but different identities are activated simultaneously, their views at the end of the cycle 
will be different and the symmetricity decreases.

Let us observe the following four properties about the algorithm:
\begin{enumerate}
\item Denote by $c$ is the center of symmetricity of the initial configuration $P_i$ which is therefore a Weber point (Lemma \ref{hom:lem:weberSym}).
Since the algorithm makes robots move only through their radius with $c$ (line \ref{A15}, the Weber point remains invariant during all the execution.
This implies that any successive regular/symmetric configuration will have necessarily $c$ as its center
of regularity/symmetricity.

\item Again, since robots move only through their radius with $c$, regularity remains invariant during all the execution.
It is thus equal to $reg(P_i)$. 
But since $sym(P_i)=m>1$, it holds according to Lemma \ref{hom:lem:regEqSym}
 that $reg(P_i)=sym(P_i)$.
Hence, all the successive configurations will be $m$-regular, including the final one.
This means that if a configuration $P^\prime$ with $sym(P^\prime)$ is reached, it is of type $\mathcal{D}(m)$.

\item At each cycle, the algorithm chooses a set $Elected$ of robots having the same view (equivalence class) (lines \ref{A10}-\ref{A13}).
These 
Since the algorithm is executed only if the current configuration is $m$-symmetric ($m>1$), it holds according
to Lemma \ref{hom:lem:sizeEquivalenceClasses} that $|Elected|=m$.
The positions chosen by robots in $Elected$ are located outside the current SEC. Hence, moving to these positions cannot increase 
the symmetricity of the configuration.
It follows that the symmetricity of the configurations can only \emph{decrease} during the execution.

\item The actions of robots maintain the same stretch during the whole execution, and it is equal to $F(P_{i+1})$ (line \ref{A5}).
\end{enumerate}

We prove the following claim about the algorithm:
\begin{lemma}
\label{hom:lem:symmetryDecrease}
Given the conditions of Lemma \ref{hom:lem:grandAlgoChiantEtDeprimant},
if robots are executing algorithm \ref{hom:alg:formation1}, then there exists a time at which they
reach a configuration $P^\prime$ with
$sym(P^\prime) < \textsc{sncd}(m, m^\prime)$.
\end{lemma}

\begin{proof}
Remember that we showed in Item 3 above that symmetricity can only decrease.
Assume towards contradiction that it remains always greater or equal to
$\textsc{sncd}(m, m^\prime)$.
This means that there exists a time $t$, a symmetricity $x > \textsc{sncd}(m, m^\prime)$, such that all the 
configurations reached after $t$ have symmetricity equal to $x$.
But we assumed that $h> \frac{n}{\textsc{sncd}(m, m^\prime)}$, which implies that $h>n/x$.
Hence, according to Lemma \ref{hom:lem:2robotsEquivDistinctLabels}, there exists 
two robots $r_1, r_2$ with identical views but distinct labels.
Let $S$ be the set of robots with the same view as $r_1, r_2$.
Note that the robots of $S$ form a regular polygon around $c$, \emph{i.e.} they lie in the same circle.
There exists a time at which the robots of $S$ are elected, \emph{i.e.} when
they become the closer to the center $c$.
Since $r_1$ and $r_2$ have distinct labels, they will choose different destinations and the symmetricity will decrease.
Contradiction.
\end{proof}

\begin{definition}[$m^\prime, T$]
\label{hom:def:mPrime}
Let $m^\prime$ be the smallest symmetricity of all the configurations reached by the execution of the algorithm.
According to Lemma \ref{hom:lem:symmetryDecrease} $m^\prime < \textsc{sncd}(m, m^\prime)$.
Let $T$ be the first reached configuration for which $sym(T)=m^\prime$.
Since all the configurations reached after $T$ if any have a symmetricity equal to $m^\prime$ this 
means that at each cycle after $T$ is reached, there are $m^\prime$ robots that are elected to move.
\end{definition}

\begin{lemma} 
\label{hom:lem:ToftypeB}

If $sym(T)=m^\prime>1$, then either $T$ is of type $\mathcal{B}(m^\prime)$ 
or a configuration of type $\mathcal{B}(m^\prime)$ can be obtained from $T$ after one cycle.
\end{lemma}

\begin{proof}
If $T$ is of type $\mathcal{B}(m^\prime)$ we are done.
If not, then $T$ is of type $\mathcal{C}(m^\prime)$ because it is $m$-symmetric.
Hence, the algorithm is executed also by robots when they are at configuration $T$.
Let $T^{+1}$ be the next configuration that is just after $T$.
By definition of $m^\prime$, $T^{+1}$ must be also $m^\prime$-symmetric. The external circle of 
$T^{+1}$ is formed by the $m^\prime$ robots that moved between $T$ and $T^{+1}$. 
They necessarily form a regular 
polygon of $m$ sides, otherwise the symmetricity would have decreased.
Moreover, the ratio between the two most external circles is greater than five, hence 
$T^{+1}$ is of type $\mathcal{B}(m^\prime)$.
This proves the lemma.
\end{proof}

\begin{lemma}
\label{hom:lem:ToftypeA}
If $m^\prime =1$ and all the points in $T$ are collinear, then either $T$ is of type $\mathcal{A}$
or a configuration of type $\mathcal{A}$ can be reached from $T$ be a movement of a single robot.
\end{lemma}

\begin{proof}
Let $T^{-1}$ denotes the configuration that just precedes $T$ in the execution.
Clearly, $T^{-1}$ is also collinear since robots move only through their radius with $c$.
Moreover, $sym(T^{-1})=2 > m^\prime=1$.
We distinguish between two subcases:
\begin{enumerate}
\item Only one robot moved between $T^{-1}$ and $T$. In this case, $T$ is of type $\mathcal{A}$
with appropriate stretch. So the algorithm stops here and the lemma follows. 

\item Two robots moved between $T^{-1}$ and $T$.
It can be easily checked that $T$ is not an intermediate configuration (not of type $\mathcal{A}$, $\mathcal{B}$, $\mathcal{C}$ or $\mathcal{D}$).
The stretch of $T$ can be however computed by trying to find the center of symmetricity of the precedent configuration $T^{-1}$. We do this by ignoring the two extreme positions in $T$ (corresponding to the robots that moved between $T^{-1}$ and $T$). $c$ is the centroid of the remaining positions.
Then, having $c$, we can deduce the stretch using the formula of configurations of type $\mathcal{D}$.
After that, we elect some robot in $T$ and make it move in such a way to form a configuration of type $\mathcal{A}$ with appropriate stretch
as showed in Lemma \ref{hom:lem:D2A}.
\end{enumerate}
\end{proof}


\begin{lemma}
\label{hom:lem:ToftypeAorD}
If $m^\prime =1$ and the points in $T$ are \emph{not} all collinear, then either $T$ is of type $\mathcal{A}$ or $\mathcal{D}(m)$
\end{lemma}

\begin{proof}
Let $T^{-1}$ denotes the configuration that just precedes $T$ in the execution.
Note that $sym(T^{-1}=2$ since $T$ is by definition the first configuration that reaches symmetricity equal to 1.
Let $S_1$ be the set of robots that moved between $T^{-1}$ and $T$ and let $S_2$ be the set of robots that remained stationary.
Note that $|S_2|\geq n/2$.
If $|S_1|=1$ then the obtained configuration $T$ is of type $\mathcal{A}$ and we are done.
Hence we assume that $|S_1|>1$ and we prove that $T$ is of type $\mathcal{D}$.
For this, it suffices to show that $T$ is neither of type $\mathcal{A}$ nor $\mathcal{B}$.
Let $u=rad(SEC(T^{-1}))$. The following two properties hold in $T$:
\begin{description}
\item[(A1)] the distance between any two robots is smaller than 
$2 \cdot (t+1)\cdot (h+1) \cdot u$. Moreover, 
\item[(A2)] $\forall r \in S_2: \forall r^\prime \in S_1: |r, r^\prime| \geq (t \cdot (h+1) - 1) \cdot u \geq (t-1) \cdot (h+1) \cdot u$.
\item[(A3)] $\forall r, r^\prime \in S_2: |r, r^\prime| \leq 2\cdot u$.
\end{description}

\paragraph*{Not $\mathcal{B}$}
Assume towards contradiction that 
$T$ is of type $\mathcal{B}$. 
$SEC_1$ and $SEC_2$ are defined accordingly. 
We say $r \in SEC_1$ if it is on $SEC_1$.
Note that at least one robot of $S_1$ must be on $SEC_1$,
and $\forall r \in T: (r \in SEC_1) \vee (r \in SEC_2)$.

\begin{claim} [C1]
$(\exists r \in S_1 \cap SEC_2) \Rightarrow  (\forall x \in S_2: x \in SEC_1)$.
\end{claim}

\begin{proof}
Fix $r$ to be some robot in $S_1 \cap SEC_2$.
Assume for contradiction that there exists some $x \in S_2 \cap SEC_2$.
By (A2) we conclude that $|x, r| \geq (t -1) \cdot (h+1)  \cdot u$.
Since both $x$ and $r$ are in $SEC_2$, its diameter $diam(SEC_2)$ must 
be greater than $(t -1) \cdot (h+1)  \cdot u$.
Hence, by definition of pattern $\mathcal{B}$, $diam(SEC_1) > 10 \cdot (t -1) \cdot (h+1)  \cdot u$.
Thus, we must have at least two robots in $SEC_1$ distant from each other by more than 
$5 \cdot (t -1) \cdot (h+1)  \cdot u$.
When $t > 2$, this contradicts (A1). This proves that $x \not\in SEC_2$.
Hence, $\forall x \in S_2: x \in SEC_1$. 
\end{proof}

\begin{claim} [C2]
$\forall r \in S_1: r \in SEC_1$ \end{claim}

\begin{proof}
Assume towards contradiction that some $r \in S_1 \cap SEC_2$.
This implies, according to C1, that $\forall x \in S_2: x \in SEC_1$. 
But $|S_2|\geq n/2$ and 
there must be at least one robot of $S_1$ in $SEC_1$. Hence this also implies
$|SEC_1|>n/2+1$, contradiction because $SEC_1$ cannot contain the majority of robots in $\mathcal{B}$. 
This proves the claim.
\end{proof}

\begin{claim}[C3]
$\forall r \in S_2: r \in SEC_2$ 
\end{claim}

\begin{proof}
Assume that some $r \in S_2 \cap SEC_1$. Observe that according to $C2$
all robots $x$ that are in $S_1$ are on $SEC_1$ also.
Hence, \textbf{(B1)} there exists at least one robot in $S_2$ that is in $SEC_2$, otherwise
the latter will be empty. Let this robot be $y$.

According to (A1), $dist(x, r) \geq (t-1) \cdot (h+1) \cdot u$.
Hence, $radius(SEC_1) \geq  (t-1) \cdot (h+1) \cdot u$.
When $t >10$ we have $radius(SEC_1) \geq 10 \cdot u$.
But $radius(SEC_1) > 10 \cdot radius(SEC_2)$.
Hence, \textbf{(B2)} if $x_1$ is in $SEC_1$ and $x_2$ is $SEC_2$
$|x_1, x_2| \geq 9 \cdot u$. 

But we have $r \in S_2 \cap SEC_1$ 
and $y \in S_2 \cap SEC_2$.
Since both $y$ and $r$ are in $S_2$, we have according to (A3),
$dist(y, r) \leq 2 \cdot u$. This contradicts (B2) and proves the claim.

\end{proof}

Hence, the robots of $SEC_1$ are those in $S_1$ and robots of $SEC_2$ are those of $S_2$.
The SEC of $S_2$ is the SEC of $T^{-1}$.
Hence, the center of $SEC_1$ is the Weber point.
According to the definition of $\mathcal{B}$,
the robots at $SEC_1$ form a regular polygon.
Hence, there is some symmetry maintained between $T^{-1}$ and $T$.
It contradicts the fact that $sym(T)=1$.
Therefore, $T$ cannot be of type $\mathcal{B}$.

\paragraph*{Not $\mathcal{A}$} can be proved using the same techniques.
\end{proof}

\paragraph*{\textbf{Proof of Lemma \ref{hom:lem:grandAlgoChiantEtDeprimant}}}

Follows from Lemma \ref{hom:lem:symmetryDecrease}, 
Definition \ref{hom:def:mPrime} and Lemmas
\ref{hom:lem:ToftypeA}, 
\ref{hom:lem:ToftypeAorD} and
\ref{hom:lem:ToftypeB}.

\begin{theorem}
\label{hom:thm:principalFormation}
Consider a robot network of $n$ robots in configuration $P_i$. Let $sym(P_i)=m>1$ and $sym(P_{i+1})=m^\prime$.
If $h> \frac{n}{\textsc{sncd}(m, m^\prime)}$, there exists an algorithm that forms $P_{i+1}$.
\end{theorem}

\begin{proof}
Follows from Lemmas \ref{hom:lem:A2anyPattern}, \ref{hom:lem:B2somePatterns} and \ref{hom:lem:grandAlgoChiantEtDeprimant}.
\end{proof}

The proof of Theorem \ref{hom:thm:formationPossibility} follows directly from Theorem \ref{hom:thm:principalFormation}.

\section{Special case: Distinct identifiers, Multiplicity points}
\label{hom:sec:threeRobots}

In this section we consider the case where $l=n$ and patterns may contain \emph{multiplicity points}.
In this context,  
we prove that
making the identifiers of robots 
invisible does not limit their computational power. That is,
the series of geometric patterns that can be formed in this case are the same that those we can form when robots 
are endowed with \emph{visible} distinct identifiers.
The following theorem states this result:

\begin{theorem}
With $n\geq 1$ robots having distinct invisible identifiers, 
we can for any finite series of distinct patterns $\langle P_1, P_2, \ldots, P_m \rangle$
iff for all $i$, $1 \leq i \leq m$, $size(P_i) \leq n$ where $size(P_i$) is the cardinality of $P_i$.
\end{theorem}

\begin{proof}
Follows from Lemmas 5.1. and 5.4 in \cite{DFSY10} and Lemma \ref{hom:lem:contradicts} below.
\end{proof}

\begin{lemma}
\label{hom:lem:contradicts}
given any non-trivial series $<P_1,P_2,...,P_m>$, where $size(P_i)> 2$ and any $l$, $1 \leq l < m$, three robots can form the following series of pattern $S$:

\begin{center}
 {$\langle P_1, P_2, ..., P_l, 
\textsc{point} , P_{l+1}, P_{l+2},...,P_{m}, \textsc{two-points} \rangle^{∞}$}\\
\end{center}
\end{lemma}

\begin{proof}


To form the series, we need to introduce some intermediate patterns denoted by $C(1,2;3)$, $C(1,3;2)$,
$C(2,3;1)$, $L_1$, $L_3$, $S_{123}$ and $S_{213}$
and whose precise definitions are given below. The obtained series that includes these patterns is the following:

\begin{center}
$\langle P_1, P_2, ..., P_l, L_1, (1,2;3), \textsc{point}, (1,3;2), S_{123}, P_{l+1}, P_{l+2},...,P_{m}, L_3, \textsc{two-points}=(2,3; 1), S_{213} \rangle ^{∞}$ \\
\end{center}

\paragraph*{}
The intermediate patterns are defined as follows:
\begin{itemize}
\item $C(i,j; l), i, j, l \in \{1, 2, 3\}$ represent a configuration in which robots $r_i$ and $r_j$ are collocated in the same point, and 
$r_l$ occupies a distinct position.
\\
\item $L_1$ (read "leader $r_1$") and $L_3$ are two scalene triangles distinct from one another and from any pattern in $\{P_1,P_2,...,P_m\}$. Moreover,
in $L_1$ (resp. $L_3$), the angle whose vertex is $r_1$ (resp. $r_3$) is the smallest one. 
\\
\item $S_{123}$ and $S_{213}$ are two configurations of the same pattern which consists in a scalene triangle with angles $\alpha_1 < \alpha_2 < \alpha_3$.
In $S_{123}$, $angle(r_i)=\alpha_i$ for any $i \in \{1, 2, 3\}$.
In contrast, in $S_{213}$ we have $angle(r_1)=\alpha_2$ and $angle(r_2)=\alpha_1$.
Note that $S_{123}$ and $S_{213}$ are distinct from the patterns $\{L_1, L_3, P_1,P_2,...,P_m\}$.
\end{itemize}

Now we explain how the transformations between different patterns are handled:

\begin{enumerate}
\item $P_i \rightsquigarrow P_{i+1}$: Note that each pattern $P_i$ consists of exactly three distinct points.
Hence the transformation between $P_i$ and $P_{i+1}$ is achieved \cite{DFSY10} by the movement of only one robot (say $r_3$).

\item $P_l \rightsquigarrow L_1$: The transformation is achieved by making $r_1$ moves towards some point such that
the obtained pattern is isomorphic to $L_1$ and $r_1$ is the vertex of its smallest angle.

\item $L_1 \rightsquigarrow (1,2; 3)$: 
Note that $r_1$ can be distinguished from other robots in $L_1$.
Hence, 
the transformation from $L_1$ to $(1,2; 3)$ is simply achieved by making $r_2$ join the location of $r_1$.

\item $(1,2; 3) \rightsquigarrow \textsc{point}$: Here, $r_3$ has to move towards the multiplicity point.

\item $\textsc{point} \rightsquigarrow  (1,3;2)$: The latter configuration is obtained by making $r_2$ move
outside the multiplicity point.

\item $(1,3;2) \rightsquigarrow S_{123}$ and $(2,3; 1) \rightsquigarrow S_{213}$: 
Both transformations are handled by $r_3$ which cannot distinguish between the two starting configurations
$(1,3;2)$ and $(2,3; 1)$. Hence, its actions in both cases are similar.
Let $p_a$ denotes the location of the multiplicity point in the starting configuration and let $p_b$ denotes
the other location.
We make $r_3$ move to some point $p_c$ such that $\sphericalangle(p_c, p_a, p_b)=\alpha_1$,
$\sphericalangle(p_c, p_b, p_a)=\alpha_2$ and $\sphericalangle(p_a, p_c, p_b)=\alpha_3$.
It can be easily checked that the obtained configuration is $S_{123}$ in the first case, and 
$S_{213}$ in the latter case.

\item $S_{123} \rightsquigarrow P_{l+1}$ and $S_{213}  \rightsquigarrow P_{1}$:
Note that the starting configurations correspond to the same pattern.
Both of the transformations are handled by $r_2$.
When $r_2$ observes a configuration consisting of a triangle with angles $\alpha_1, \alpha_2$ and $\alpha_3$,
it knows that the current configuration is either $S_{123}$ or $S_{213}$. 
It can distinguish between them by looking to its own angle (equal to $\alpha_2$ in the first case and
to $\alpha_1$ in the second one). Hence, it can execute the appropriate transformation towards $P_{l+1}$ or $P_1$.

\item $P_{m} \rightsquigarrow L_3$: Similar to 2.

\item $L_3 \rightsquigarrow (2,3; 1)$: Similar to 3.

\end{enumerate}
\end{proof}

\section{Conclusion}
\label{hom:sec:conclusion}
In this paper, we considered the problem of formation of series of geometric patterns.
We studied the combined effect of obliviousness and anonymity on the computational power of mobile robots
with respect to this problem.
To this end, we introduced a new model, robots networks with homonyms that encompasses and generalizes
the two previously considered models in the literature in which either all robots have distinct identifiers or they are anonymous.
Our results suggest that this new model may be a useful tool to get a better insight on how 
anonymity interacts with others characteristics of the model to limit its power.

\section*{Acknowledgement}
The authors thank Dr. Shantanu Das for his generous help and helpful comments.

\bibliographystyle{plain}
\bibliography{homonyms}

\begin{thebibliography}{10}

\bibitem{agmon2004ftg}
N.~Agmon and D.~Peleg.
\newblock {Fault-tolerant gathering algorithms for autonomous mobile robots}.
\newblock {\em SODA}, 11(14):1070--1078, 2004.

\bibitem{ACP05}
L.~Anderegg, M.~Cieliebak, and G.~Prencipe.
\newblock Efficient algorithms for detecting regular point configurations.
\newblock {\em Theoretical Computer Science}, pages 23--35, 2005.

\bibitem{CT90}
R.~Chandrasekaran and A.~Tamir.
\newblock Algebraic optimization: the fermat-weber location problem.
\newblock {\em Mathematical Programming}, 46(1):219--224, 1990.

\bibitem{DFSY10}
S.~Das, P.~Flocchini, N.~Santoro, and M.~Yamashita.
\newblock On the computational power of oblivious robots: forming a series of
  geometric patterns.
\newblock In {\em PODC}, pages 267--276. ACM, 2010.

\bibitem{DFGKR11}
C.~Delporte-Gallet, H.~Fauconnier, R.~Guerraoui, A.M. Kermarrec, E.~Ruppert,
  and H.~Tran-The.
\newblock Byzantine agreement with homonyms.
\newblock {\em PODC}, 2011.

\bibitem{FIPS10}
P.~Flocchini, D.~Ilcinkas, A.~Pelc, and N.~Santoro.
\newblock Remembering without memory: Tree exploration by asynchronous
  oblivious robots.
\newblock {\em Theoretical Computer Science}, 411(14-15):1583--1598, 2010.

\bibitem{K05}
B.~Katreniak.
\newblock Biangular circle formation by asynchronous mobile robots.
\newblock {\em Structural Information and Communication Complexity}, pages
  185--199, 2005.

\bibitem{KMP77}
D.E. Knuth, J.H. Morris~Jr, and V.R. Pratt.
\newblock Fast pattern matching in strings.
\newblock {\em SIAM journal on computing}, 6:323, 1977.

\bibitem{SY99}
I.~Suzuki and M.~Yamashita.
\newblock Distributed anonymous mobile robots: Formation of geometric patterns.
\newblock {\em SIAM Journal of Computing}, 28(4):1347--1363, 1999.

\bibitem{W37}
E.~WEISZFELD.
\newblock Sur le point pour lequel la somme des distances de n points donnes
  est minimum, t6hoku math.
\newblock {\em J}, 43:355--386, 1937.

\end{thebibliography}

\appendix
\end{document}